\documentclass{article}
\usepackage[english]{babel}
\usepackage{graphicx,multicol}
\usepackage{epic,eepic,epsfig}
\usepackage{amssymb}
\usepackage{amsmath,amsthm}
\usepackage{color}
\usepackage{tikz}
\usepackage{epsfig,url}
\usepackage{multirow}
\usepackage{amsfonts}
\usepackage{setspace}
\usepackage{geometry}
\usetikzlibrary{shapes,snakes}

\usepackage{eso-pic} % Magic that gets rid the newline errors in the pics

\newcommand{\induce}[2]{\mbox{$ #1 [ #2 ]$}}
\newcommand{\cs}{\JBJ{ss}}

\newcommand{\dom}{\mbox{$\rightarrow$}}

\newcommand{\be}{\begin{enumerate}}
\newcommand{\ee}{\end{enumerate}}
\newcommand{\bd}{\begin{description}}
\newcommand{\ed}{\end{description}}

\newcommand{\beq}{\begin{equation}}
\newcommand{\eeq}{\end{equation}}

\newcommand{\2}{\vspace{2mm}}

\renewenvironment{proof}[1][]{\par \noindent {\bf Proof#1}.\ }{\hfill$\Box$
\par \vspace{11pt}}

\newtheorem{theorem}{Theorem}[section]
\newtheorem{lemma}[theorem]{Lemma}
\newtheorem{proposition}[theorem]{Proposition}
\newtheorem{corollary}[theorem]{Corollary}
\newtheorem{claim}{Claim}

\theoremstyle{definition}

\newtheorem{conjecture}[theorem]{Conjecture}
\newtheorem{question}[theorem]{Question}

\newtheorem{fact}{Fact}
\newtheorem{case}{Case}

\newcommand{\pf}{{\bf Proof: }}
\newcommand{\JBJ}[1]{{#1}}
\newcommand{\AY}[1]{{#1}}

\begin{document}
\bibliographystyle{plain}
%\pagewiselinenumbering
%\setpagewiselinenumbers
%\modulolinenumbers[1]
%\linenumbers

\begin{spacing}{1.2}

\title{Safe sets in digraphs}
\author{Yandong Bai\thanks{Department of Applied Mathematics, Northwestern Polytechnical University, Xi'an, Shaanxi Province 710129, China
Xi’an-Budapest Joint Research Center for Combinatorics, Northwestern Polytechnical University, Xi'an, Shaanxi Province 710129, China.
Supported by NSFC (No. 11601430) and  the Fundamental Research Funds for the Central Universities (No. 3102019ghjd003)
Email: bai@nwpu.edu.cn}, J\o{}rgen Bang-Jensen\thanks{Department of Mathematics and Computer Science, University of Southern Denmark. Research supported by the Danish research council for independent research under grant number DFF 7014-00037B. Email: jbj@imada.sdu.dk}, Shinya Fujita\thanks{School of Data Science, Yokohama City University, 22-2, Seto, Kanazawa-ku, Yokohama, 236-0027 Japan.
Supported by JSPS KAKENHI  (19K03603). 
Email: fujita@yokohama-cu.ac.jp}, Anders Yeo\thanks{Department of Mathematics and Computer Science, University of Southern Denmark. Research supported by the Danish research council for independent research under grant number DFF 7014-00037B. Email: yeo@imada.sdu.dk}}

\date{\today}
\maketitle

\begin{abstract}
A non-empty subset $S$ of the vertices of a digraph $D$ is called a {\it safe set}
if
\begin{itemize}
  \item[(i)] for every strongly connected component $M$ of $D-S$,
there exists a strongly connected component $N$ of $D[S]$
such that there exists an arc from $M$ to $N$; and \item[(ii)] for every strongly connected component $M$ of $D-S$
and every strongly connected component $N$ of $D[S]$,
we have $|M|\leq |N|$ whenever there exists an arc from $M$ to $N$.
\end{itemize}
\JBJ{In the case of acyclic digraphs a set $X$ of vertices is a safe set precisely when $X$ is an {\it in-dominating set}, that is, every vertex not in $X$ has at least one arc to $X$. We prove that, even for acyclic digraphs which are traceable (have a hamiltonian path) it is NP-hard to find a minimum cardinality in-dominating set.  Then we show that the problem is also NP-hard for tournaments and give, for every positive constant $c$, a polynomial algorithm for finding a minimum cardinality safe set in a tournament on $n$ vertices in which no strong component has size more than $c\log{}(n)$. Under the so called Exponential Time Hypothesis (ETH) this is close to best possible in the following sense: If ETH holds, then, for every $\epsilon>0$ there is no polynomial time algorithm for finding a minimum cardinality safe set for the class of tournaments in which the largest strong component has size at most \AY{$\log^{1+\epsilon}(n)$}.
  We also discuss bounds on the  cardinality of safe sets in tournaments.} \\
\noindent{}{\bf Keywords: Safe set, tournament, in-dominating set, NP-complete, polynomial algorithm}
\end{abstract}

\section{Introduction}

We use Bang-Jensen and Gutin \cite{BG2008} for terminology and notation not defined here.
Only finite, simple graphs and digraphs are considered. For a digraph $D$ we denote by $|D|$ the order of $D$, that is the number of vertices in $D$.

If $D=(V,A)$ is a digraph and $X\subset V$, then we denote by $D[X]$ the subdigraph induced by the vertices in $X$. If $U,W$ are disjoint subsets of the vertex set of a digraph $D=(V,A)$ such that there is an arc $uw\in A$ for every  $u\in U,w\in W$, then we denote this by $U\dom W$.

A digraph $D$ is {\it strongly connected} or {\it strong}
if there exists a directed path from $u$ to $v$ for any two vertices $u$ and $v$ of $D$,
and $D$ is {\it $k$-strong} if the removal of any set of fewer than $k$ vertices
results in a strongly connected digraph.
The strong connectivity of $D$, denoted $\kappa{}(D)$,  is the maximum  $k$ such that $D$ is $k$-strong.
In particular, a non-strongly connected digraph has strong connectivity 0.
A {\it separator} of a strong digraph $D=(V, A)$ is a proper subset $X\subset V$ such that the digraph $D-X$, obtained by deleting the vertices of $X$ and all incident arcs, is not strong. 

An {\it oriented graph} is a digraph without directed cycles of length 2. Note that every $k$-strong digraph
has both minimum out-degree and minimum in-degree at least $k$.
So every $k$-strong $n$-vertex oriented graph has $n\geq 2k+1$ and $k\leq \lceil n/2 \rceil -1$.
From this and the well-known fact that there are $k$-strong tournaments on $2k+1$ vertices (see also Section~\ref{tournament}), we get the following \JBJ{easy fact.}

\begin{proposition}\label{kappaatmost}
The connectivity of an $n$-vertex oriented graph
is at most $\lfloor\frac{n-1}{2} \rfloor$
and \JBJ{for every integer $n\geq 1$} there exists an $n$-vertex oriented graph
with connectivity $\lfloor \frac{n-1}{2} \rfloor$.%,
%that is $\mathcal{T}^{\lfloor \frac{n-1}{2} \rfloor}_{n}\neq \emptyset$.
\end{proposition}

A non-empty subset $S$ of the vertices of a  connected \AY{undirected} graph $G$
is a {\it safe set} if,
for every connected component $M$ of $G-S$
and every connected component $N$ of $G[S]$,
we have $|M|\leq |N|$ whenever there exists an edge of $G$ between $M$ and $N$.
\AY{Note that as $G$ is a connected undirected graph this implies that every component of $G-S$ has an edge to $S$.}
For a connected graph $G$, the {\it safe number} of $G$ \JBJ{is} the minimum cardinality of \JBJ{a safe set} in $G$. \JBJ{This is well-defined since for any graph $G$ the set of all vertices form a safe set of $G$.}

The notion of safe sets was originally introduced by Fujita et al. \cite{FMS2016} as a variation of facility location problems. Kang et al. \cite{kkp} explored the safe number of the Cartesian product of two complete graphs, and Fujita and Furuya \cite{SF2018} found a close relationship between the safe number and the integrity of a graph.
Bapat et al. \cite{wsf} extended the notion of safe sets to the weighted version in vertex-weighted graphs $(G, w)$.

Motivated by some real applications such as network vulnarability, weighted safe set problems in graphs \JBJ{have received much attention, especially algorithmic aspects of safe sets}.
Fujita et al \cite{FMS2016} showed that computing the connected safe number in the case $(G, w)$ with a constant weight function $w$ is NP-hard in general. On the other hand,
when $G$ is a tree and $w$ is a constant weight function, they constructed a linear time algorithm for computing the connected safe number of $G$. \'{A}gueda et al. \cite{safe2018} constructed an efficient algorithm for computing the safe number of an unweighted graph with bounded treewidth. Somewhat surprisingly, Bapat et al. \cite{wsf} showed that computing the connected weighted safe number in a tree is NP-hard even if the underlying tree is restricted to be a star. They also constructed an efficient algorithm computing the safe number for a weighted path. Furthermore, Fujita et al. \cite{FPS:path:cycle} constructed a linear time algorithm computing the safe number for a weighted cycle.
Along a slightly different line, Ehard and Rautenbach \cite{ehard} gave a polynomial-time approximation scheme (PTAS) for the connected safe number of a weighted tree.
Very recently, the parameterized complexity of safe set problems was investigated by Belmonte et al. \cite{b} and a mixed integer \JBJ{linear programing formulation} for safe sets was introduced by Hosteins \cite{host}.  Thus, a lot of work has been done so far in this area of study.

In this paper, we consider the directed version of safe sets. 
A non-empty subset $S$ of the vertices of a digraph $D$ is a {\it safe set}
if the following two conditions hold:
\begin{itemize}
  \item [(i)] For every strongly connected component $M$ of $D-S$,
there exists a strongly connected component $N$ of $D[S]$
such that there exists an arc from $M$ to $N$;
  \item [(ii)] For every strongly connected component $M$ of $D-S$
and every strongly connected component $N$ of $D[S]$,
we have $|M|\leq |N|$ whenever there exists an arc from $M$ to $N$.
\end{itemize}
Moreover, if $D[S]$ is strongly connected,
then $S$ is called a {\it strong safe set} of $D$.
Observe that every digraph $D$ has a trivial safe set $S=V(D)$.
%and every strongly connected digraph $D$ has a trivial strong safe set $S=V(D)$.

The {\it safe number} $s(D)$ of a digraph $D$ is defined as
\begin{equation}
s(D)=\min\{|S|: S~is~a~safe~set~of~D\},
\end{equation}
and the {\it strong safe number} $\cs{}(D)$ of $D$ is defined as
\begin{equation}
\cs{}(D)=\min\{|S|: S~is~a~strong~safe~set~of~D\}.
\end{equation}
Note that a non-strong digraph may not have a strong safe set,
for example, a directed path of length at least $2$ has no strong safe set.
For convenience,
we define $\cs{}(D)=+\infty$ if $D$ has no strong safe set.
Observe that if the underlying graph of $D$ is not connected then $\cs{}(D)=+\infty$.
Observe also that the following fact holds.

\begin{proposition}
Let $D$ be a digraph. Then $s(D)\leq \cs{}(D)$.
\end{proposition}

\JBJ{In the case when $D=(V,A)$ is an acyclic digraph a subset $X\subseteq V$ is a  safe set if and only if $X$ is an {\it in-dominating set}, that is, every vertex of $V-X$ has an arc to $X$. For any digraph $D$ we denote by $\gamma{}(D)$ the size of a minimum in-dominating set. The {\sc safe set} problem on digraphs is as follows: Given a digraph $D$ and a natural number $k$; decide whether $D$ has a safe set of size at most $k$.}

In this paper, we shall initiate the study on safe sets in digraphs. As an initial step, we will mainly \JBJ{focus on safe sets in acyclic digraphs and tournaments. 
The paper is organized  as follows.
In Section~\ref{complsec} we give  a simple proof of the well-known fact that it is NP-hard to find a minimum cardinality in-dominating (safe) set in an acyclic digraph and  show that we can find a minimum in-dominating (safe) set in polynomial time for an acyclic digraph with bounded independence number.  
In Section~\ref{DAGsec}, we show that, the problem for finding a minimum cardinality
  in-dominating (safe) set is NP-hard even if the input is an acylic digraph which has a hamiltonian path (and hence has a unique acyclic ordering\footnote{An {\it acyclic ordering} of a digraph $D=(V,A)$ is an ordering $v_1,\ldots{},v_n$ of its vertices such that every arc $v_iv_j\in A$ satisfies that $i<j$.}).
  In Section \ref{Talg} we study the complexity of {\sc safe set} for tournaments and semicomplete digraphs. We show that {\sc safe set} is NP-complete for tournaments and give a polynomial dynamic programming based method that finds a minimum cardinality safe set in any tournament whose maximum size strong component is at most a constant times the logarithm of its order. We also show that, under the well-known Exponential Time Hypothesis (ETH), there is no polynomial algorithm for finding a minimum cardinality safe set in a tournament whose largest component may be slightly larger than logarithmic.
In Section~\ref{tournament}, we discuss upper and lower bounds on the safe number and the strong safe number in tournaments. 
Finally, in  Section~\ref{last} we discuss some open problems and possible directions for further research}.  

\section{Complexity of finding safe sets \JBJ{in acyclic digraphs}}\label{complsec}

\JBJ{Recall that when} the digraph in question is acyclic, the definition of a safe set coincides with that of an {\it in-dominating set}. Even for acyclic digraphs \JBJ{\sc safe set is NP-complete}. This is well-known (under the name of in-dominating set), see e.g. \cite{ganianDAM168}, but we give a proof here since it is very short and we will refer to it in Section \ref{DAGsec}. The \JBJ{{\sc set cover}} problem is the following: given a collection $S_1,S_2,\ldots{},S_m$ of subsets of a ground set $X=\{x_1,x_2,\ldots{},x_n\}$ and a non-negative integer $k$; decide if there exists a subset $Z\subseteq X$ such that
$Z\cap S_i\neq\emptyset$ for $i\in [m]$ and $|Z|\leq k$. This is one of Karp's 21 NP-complete problems in \cite{karp1972}. To reduce this problem to the in-dominating set problem, we construct a digraph $D$ with vertices $s_1,\ldots{},s_m,v_1,\ldots{},v_n,z$ and add an arc from $s_i$ to $v_j$ whenever $x_j\in S_i$. Finally we add the arc $v_iz$ for $i\in [n]$. It is easy to see that $D$ has an in-dominating set of size at most $k+1$ if and only if there is a subset $Z\subseteq X$ of size $k$ that intersects every $S_i$, $i\in [m]$.

\iffalse
In view of the above observation, we see that the problem of finding safe sets in acyclic digraphs is essentially equivalent to the problem of finding in-dominating sets in acyclic digraphs and therefore, the decision problem on the existence of a minimum safe set of size at most $k$ in an acyclic digraph seems to be NP-complete. To see that this is the case, as we will consider in-dominating sets in acyclic digraphs rather than safe sets, we will discuss further details in the next section. \fi

\JBJ{The class of acyclic digraphs that we construct  from  set cover instances above has unbounded independence number. We now  show that {\sc safe set} is polynomial for acyclic digraph of bounded independence number.} 
\AY{In order to prove this we need the following theorem due to Gy\'arf\'as et al. \cite{GST}.}

\begin{theorem}[Gy\'arf\'as et al. (Proposition 5; \cite{GST})]
\label{indep}
If $D$ is an acyclic digraph with independence number $\alpha$, then $\gamma(D)\le \alpha$.
\end{theorem}

\AY{ The following result now follows easily.}

\begin{theorem}\label{alpha}
  \JBJ{For every fixed natural number $\alpha$ {\sc safe set} is polynomial for the class of acyclic digraphs with independence number at most $\alpha$.}

\end{theorem} 

\pf 
\JBJ{ Theorem~\ref{indep} implies that $D$ contains an in-dominating set $S$ of size at most $\alpha$. 
Since $D$ is an acyclic digraph, $S$ is also a safe set in $D$. Thus, to find out a minimum safe set in $D$, we just have  to check all the vertex subsets of size at most $\alpha$ in $V(D)$. This can be done in time $O(n^{\alpha})$.}   \qed

%\section{Properties of safe sets in tournaments}
\section{In-dominating sets in \JBJ{traceable} acyclic digraphs}\label{DAGsec}

\iffalse We now turn to acyclic digraphs and in-dominating sets in these. We saw in Section \ref{complsec} that the problem of finding a minimum size in-dominating set is $\cal NP$-hard already for acyclic digraphs.
\fi
The acyclic digraph used in the reduction from \JBJ{{\sc set cover} to {\sc in-dominating set}}  above has very few arcs and no long directed paths. Furthermore, finding a minimum in-dominating set is trivial for acyclic tournaments (these are the transitive tournaments), where the minimum is always one. Hence one might think that the problem could be polynomial for acyclic digraphs that share some of the structure of the transitive tournaments. One such property is that of having a unique acyclic ordering.%, that is, an ordering $v_1,v_2,\ldots{},v_n$ so that all arcs go forward with respect to the ordering.
We shall prove below that even in the more restricted case where the acyclic digraph has a hamiltonian path, the problem is still $\cal NP$-hard. To do so we first need a few results on special classes of satisfiability problems.

For a given instance $\cal F$ of \AY{{\sc $k$-SAT}} with variables $x_1,x_2,\ldots{},x_n$ and clauses $C_1,C_2,\ldots{},C_m$  we define the graph $G=G({\cal F})$ as follows:
$G$ has one vertex $c_j$ for each clause $C_j$, $j\in [m]$ and one vertex $v_i$ for each variable $x_i$, $i\in [n]$. There is an edge between $c_i$ and $c_j$ if these share a literal and there is an edge between $v_i$ and $c_j$ if the variable $x_i$ appears in a literal of $C_j$ (either as $x_i$ or as $\bar{x}_i$). We say that an instance $\cal F$ of \AY{{\sc $k$-SAT}} is {\it traceable} if the graph $G({\cal F})$ has a hamiltonian path.

The {\it bipartite incidence graph} $B({\cal F})$ of a \AY{{\sc $k$-SAT}} instance $\cal F$ with $n$ variables and $m$ clauses is the bipartite graph which has one vertex $v_i$  for each  variable $x_i$ and one vertex $c_j$ for each  clause $C_j$ and an edge from $v_i$ to $c_j$ precisely when $C_j$ contains a literal over $x_i$. Hence $B({\cal F})$ is a subgraph of $G({\cal F})$.

\begin{lemma}
\label{reducibleformula}
If $B({\cal F})$ does not have a matching covering every $v_i$, then we can delete some clauses from $\cal F$ to obtain an equivalent instance ${\cal F}'$.
\end{lemma}

\pf Suppose there is no matching covering $\{v_1,v_2,\ldots{},v_n\}$. Then by Hall's theorem, there exists a subset $U\subseteq \{v_1,v_2,\ldots{},v_n\}$ such that $|N(U)|<|U|$, where $N(U)=\{c_{i_1},\ldots{},c_{i_p}\}$ is the set of clauses containing a literal on at least one of the variables in $U$. Assume that $U$ is chosen among all such non-matcheable sets so that $r=|U|-|N(U)|$ is maximized.
Then, it follows from Hall's theorem applied to the bipartite subgraph induced by $U\cup N(U)$ that this has a matching meeting all vertices of $\{c_{i_1},\ldots{},c_{i_p}\}$. Thus we can satisfy all clauses in  $\{C_{i_1},\ldots{},C_{i_p}\}$ using only the variables in $U$ and these variables cannot be used to 
satisfy any clause $C_j\in \{C_1,C_2,\ldots{},C_m\}-\{C_{i_1},\ldots{},C_{i_p}\}$ since they do not appear as a literal in $C_j$. This means that the formula ${\cal F}'$ consisting of the variables of $\{v_1,\ldots{},v_n\}-U$ and the clauses
$\{C_1,C_2,\ldots{},C_m\}-\{C_{i_1},\ldots{},C_{i_p}\}$ is satisfiable if and only if $\cal F$ is satisfiable. \qed

Call a \AY{{\sc $3$-SAT}} instance $\cal F$ {\it irreducible} if  $B({\cal F})$ has  a matching covering every variable vertex $v_i$. By Lemma \ref{reducibleformula}, 3-SAT is still NP-complete when the input is an irreducible instance.

\begin{lemma}
\label{HampathSAT}
\AY{{\sc $4$-SAT}} is NP-complete even when the instance is traceable.
\end{lemma}

\pf Let ${\cal F}'$ be an arbitrary irreducible instance of \AY{{\sc $3$-SAT}} with variables $x_1,\ldots{},x_{n'}$ and clauses $C'_1,\ldots{},C'_{m'}$. By adding new variables or clauses we may assume that both $n'$ and $m'$ are even. We form the \AY{{\sc $4$-SAT}} instance $\cal F$ with variables $x_1,\ldots{},x_{n'},x_{n'+1}$, where $x_{n'+1}$ is a new variable and forming two clauses $C_{j,1},C_{j,2}$ for each clause $C'_j$ of ${\cal F}'$, where $C_{j,1}$ is obtained  by adding the literal $x_{n'+1}$ to 
$C'_j$ and $C_{j,2}$ is obtained by adding the literal $\bar{x}_{n'+1}$ to $C'_j$.  Clearly ${\cal F}'$ is satisfiable if and only if $\cal F$ is satisfiable so we just have to argue that $G({\cal F})$ has a hamiltonian path. To see this, first note that, by the definition of $G({\cal F})$,  the set of all clauses containing the literal $w$ ($\bar{w}$) induce a clique in $G({\cal F})$ and for $j\in [m']$ the vertices $c_{j,1},c_{j,2}$ are adjacent (they share 3 literals). Furthermore, $G({\cal F})$ contains two  isomorphic copies of $B({\cal F}')$. These are induced by the vertices 
$\{v_1,\ldots{},v_{n'}\}\cup \{c_{1,1},\ldots{},c_{m',1}\}$ and 
$\{v_1,\ldots{},v_{n'}\}\cup \{c_{1,2},\ldots{},c_{m',2}\}$, respectively.
Now it follows from the fact that ${\cal F}'$ is an irreducible instance that
there are  two matchings $M_1,M_2$ both of size $n$ which cover  $\{v_1,\ldots{},v_{n'}\}$ and so that $M_i$ matches this set to  $\{c_{1,i},\ldots{},c_{m',i}\}$ for $i=1,2$. Furthermore we can assume that if $v_i$ is matched to $c_{j_i,1}$ in $M_i$, then it is matched to $c_{j_i,2}$ in $M_2$. By renumbering the clauses in $\cal F$ if necessary,
we may assume that $c_{j,i}$ is matched to $v_j$ in $M_i$ for $j\in [n']$. Finally, we may also assume that $m'>n'+1$ and now we get a hamiltonian path of $G({\cal F})$ by taking the path 
$c_{1,1}v_1c_{1,2}c_{2,2}v_2c_{2,1}c_{3,1}v_3c_{3,2}\ldots{}c_{{n'},2}v_{n'}c_{{n'},1}c_{{n'}+1,1}v_{n'+1}c_{{n'}+1,2}c_{{n'}+2,2}c_{{n'}+2,1}\ldots{}c_{m-1,1}c_{m-1,2}c_{m,2}c_{m,1}$.\\
 Here we used the assumption that both ${n'}$ and $m'$ are even.
\qed

\begin{theorem}
It is NP-complete to decide for input $D,k$ where $D$ is a traceable acyclic digraph and $k$ is an integer, whether $D$ has an in-dominating set of size at most $k$.
\end{theorem}

\pf Let $\cal F$ be a traceable instance of \AY{{\sc $4$-SAT}} with variables $x_1,x_2,\ldots{},x_n$ and clauses $C_1,C_2,\ldots{},C_m$ and let $P$ be a hamiltonian path of $G({\cal F})$ which starts in the vertex corresponding to $C_1$ and ends in the vertex corresponding to $C_m$ (we may choose the numbering of the clauses so that this holds). We may also assume, by relabeling the variables and negating all clauses if necessary, that $x_1$ appears as the  literal $x_1$ in $C_m$.
We build a traceable  acyclic digraph $D$ which has an in-dominating set of size $n+1$ if and only if $\cal F$ is satisfiable. The vertex set of $D$ is $\{v_1,\ldots{},v_n\}\cup \{c_1,\ldots{},c_m\}\cup\{w_1,\ldots{},w_n\}\cup\{\bar{w}_1,\ldots{}\bar{w}_n\}\cup\{u\}$ where the vertex $c_j$ corresponds to the clause $C_j$, $j\in [m]$ and the three vertices $v_i,w_i,\bar{w}_i$ correspond to the variable $x_i$, $i\in [n]$. Here $w_i$ corresponds to the literal $x_i$ and $\bar{w}_i$ corresponds to the literal $\bar{x}_i$. Now we add the following arcs: 

\begin{itemize}
\item Start with no arcs.
\item First add the arcs $\{v_iw_i,v_i\bar{w}_i: i\in [n]\}\cup\{w_iu,\bar{w}_iu: i\in [n]\}$.
For $j\in [m]$ and each of the 4 literals of $C_j$ add an arc from $c_j$ to the vertices corresponding to these literals, e.g. if $C_j=(x_3\vee\bar{x}_5\vee\bar{x}_7\vee x_9)$ then we add the arec $c_jw_3,c_j\bar{w}_5,c_j\bar{w}_7,c_jw_9$.
\item Note that the vertices $\{c_1,\ldots{},c_m\}\cup\{v_1,\ldots{},v_n\}$ correspond to those of $G({\cal F})$. Add the arcs of $P$ oriented as a directed path $Q$ from $c_1$ to $c_m$ (recall that this path contains all the vertices in  $\{c_1,\ldots{},c_m\}\cup\{v_1,\ldots{},v_n\})$.
\item Add the arcs of the directed path $Q'=w_1\bar{w}_1w_2\bar{w}_2\ldots{}w_n\bar{w}_n$.
\end{itemize}

First note that $P^*=QQ'u$ is a hamiltonian path of $D$ and that $D$ is acyclic (all arcs go forward in the ordering inducted by $P^*$). We will now show that $D$ has an in-dominating set of size $n+1$ if and only if $\cal F$ is satisfiable.

Suppose first that $\phi$ is a satisfying truth assignment for $\cal F$ then let $Z$ be a subset of $V(D)$ defined as follows: for each $i\in [n]$ if $\phi{}(x_i)=`true'$, then $Z$ contains $w_i$ and if $\phi{}(x_i)=`false'$ then $Z$ contains $\bar{w}_i$. Finally $Z$ contains $u$. Now let us see that $Z$ is an in-dominating set. Note that, for each $i\in [n]$ we add the vertex corresponding to the true literal among $x_i$ and $\bar{x}_i$ to $Z$. This means that every vertex of $\{c_1,\ldots{},c_m\}$ has an out-neighbour in $Z$. Clearly this also holds for the vertices in $\{v_1,\ldots{},v_n\}$. Finally, every vertex of $\{w_i,\bar{w}_i: i\in [n]\}$ dominates $u$ so $Z$ is in-dominating.

Now suppose that $D$ has an in-dominating set of size $n+1$. It is clear that $u$ is in every in-dominating set as it has out-degree zero. We claim that $D$ also contains an in-dominating set $Z$ of size $n+1$ such that $Z\subset \{w_i,\bar{w}_i: i\in [n]\}\cup \{u\}$.  This follows from the way we defined adjacencies in the graph $D$: Let $Z$ be an in-dominating set of size $n+1$ which uses the minimum number of vertices from $\{c_1,\ldots{},c_m\}\cup \{v_1,\ldots{},v_n\}$. Recall from the definition of $D$ that the maximum semi-degree of a vertex in $D'=\induce{D}{\{c_1,\ldots{},c_m\}\cup \{v_1,\ldots{},v_n\}}$ is one ($D'$  is an induced  hamiltonian directed path). Suppose that $Z$ contains a vertex $c_i$ and let $q$ be the unique in-neighbour of $c_i$ in $D'$. Then $c_i$ and $q$ have a common out-neighbour $y$ in $\{w_j,\bar{w}_j\}$ for some $j\in [n]$, implying that we could replace $c_i$ by $y$ and still have an in-dominating set, contradicting the choice of $Z$. It is easy to see that the same contradiction can be reached if $Z$ contained a vertex $v_a$ for some $a$. Thus we have shown that $Z\subset\{w_i,\bar{w}_i: i\in [n]\}\cup \{u\}$. Now it follows from the fact that each $v_i$ has precisely two out-neigbours in $\{w_i,\bar{w}_i: i\in [n]\}\cup \{u\}$ that  $Z$ must contain precisely one of the vertices $w_i,\bar{w}_i$ for each $i\in [n]$. Thus setting $x_i$ 'true' if $w_i\in Z$ and 'false' if
$\bar{w}_i\in Z$, every clause will contain a true literal so we obtain a satisfying truth assignment to $\cal F$. \qed.

%\begin{problem}
%For which classes of acyclic digraphs can we find a minimum safe set in polynomial time?
%\end{problem}

%\begin{problem}
%Is minimum safe set polynomial in acyclic digraphs with independence number 2? 
%\end{problem}

\begin{corollary}
It is NP-complete to decide for input $D,k$ where $D$ is a traceable acyclic digraph and $k$ is an integer, whether $D$ has a safe set of size at most $k$.
\end{corollary}

\section{Algorithmic aspects of {\sc safe set} for  tournaments}\label{Talg}

\iffalse In this section we consider the complexity of finding a minimum safe set in tournaments.
\fi

\2

\begin{theorem}\label{thm:NPhard_tournament}
  % For tournaments $T$, it is NP-hard to determine $s(T)$.
  {\sc Safe Set} is NP-complete for tournaments.
\end{theorem}

\JBJ{
  \begin{proof}
    Recall that a {\it feedback vertex set} in a digraph $D=(V,A)$ is a subset $X\subseteq V$ such that $\induce{D}{V-X}$ is acyclic.
    The {\sc Feedback vertex set} problem asks for a given digraph $D$ and a natural number $k$ whether $D$ has a feedback vertex set of size at most $k$. This problem is on Karp's list of NP-complete problems from 1972
    \cite{karp1972} and was shown to be NP-complete already for tournaments in
    \cite{speckenmeyerLNCS411} (see also \cite{bangSJDM5}). Given an instance \AY{$(T,k)$} of {\sc Feedback vertex set} for tournaments. We construct a new tournament $T'$ by adding a new vertex $x$ and all possible arcs from  \AY{$V(T)$} to $x$. The tournament $T'$ has at least two strong components and every safe set of $T'$ must contain the vertex $x$ and a subset $W\subset V$ such that $\induce{T}{V-W}$ is acyclic. This follows from the fact that every vertex of $V$ has an arc to $x$ in $T'$. Thus $T'$ has a safe set of size $k+1$ if and only if $T$ has a feedback vertex set of size $k$, implying that {\sc Safe Set} is NP-complete for tournaments.
    \end{proof}

  }

  \JBJ{Let  ${\cal T}(f(n))$  be the class of tournaments, $T$, satisfying that  the size of the largest strong component of $T$ is at most $f(|T|)$.}

Using the so-called sparsification lemma, the statement of ETH (the Exponential Time Hypothesis) is often formulated as follows (see e.g. Theorem 14.4 in \cite{cygan2016}):

\2

{\bf Exponential Time Hypothesis (ETH):} {\em Unless ETH fails, there is a constant $\delta > 0$, such that no algorithm solves \AY{{\sc 3-SAT}} in 
time $O(2^{\delta (n+m)})$, where $n$ denotes the number of variables and $m$ the number of clauses. }

\begin{theorem}\label{thm:small_strong_comp_ETH_tournament}
  If ETH holds and $\varepsilon >0$, then there is no algorithm that for every  $T \in {\cal T}\left(\log^{(1+\varepsilon)}(n) \right)$ will  
determine $s(T)$ in polynomial time.
\end{theorem}

\JBJ{
  \begin{proof}
    Recall that {\sc Clique} is the problem of deciding, for input $G$ and $k$, whether the undirected graph $G$ has a complete subgraph of size $k$ and that {\sc Vertex  Cover} is the problem of deciding, for input $G=(V,E)$ and $k$, whether the undirected graph $G$ has a subset
$X\subseteq V$ such that every edge $e\in V$ has at least one of its end vertices in $X$.
 Now consider the following chain of polynomial reductions
    {\sc 3-SAT} $\leq_p$ {\sc Clique} $\leq_p$ {\sc Vertex  Cover}, see e.g. \cite[Section 34.5]{cormen2009}.
  Through these two reductions an instance $\cal F$ of {\sc 3-SAT} with $n$ variables and $m$ clauses is transformed into an instance\AY{, $G$,} of {\sc Vertex  Cover} with $3m$ vertices so that $G$ has a vertex cover of size at most $2m$ if and only if $\cal F$ is satisfiable. In \cite{speckenmeyerLNCS411} a polynomial reduction from {\sc Vertex  Cover} to {\sc Feedback vertex set} for tournaments is given. In this reduction an instance of {\sc vertex cover} with $n$ vertices and parameter $k$ is transformed into an instance of {\sc Feedback vertex set} for tournaments with $3n$ vertices and parameter $k$. Putting this together with the reductions above, we have a polynomial reduction from {\sc 3-SAT} to 
  {\sc Feedback vertex set} for tournaments which transforms an instance $\cal F$ of \AY{{\sc 3-SAT}} with $n$ variables and $m$ clauses into an instance $T$ of {\sc Feedback vertex set} for tournaments on $9m$ vertices such that $T$ has a feedback vertex set of size $2m$ if and only if $\cal F$ is \AY{satisfiable.

% 3-SAT is NP-complete even if no variable appears as a literal in more than 4 clauses \cite{toveyDAM8}. 

Now} let $\epsilon >0$ be given and suppose that there is an  algorithm $\cal A$ that for every  $T \in {\cal T}\left(\log^{(1+\varepsilon)}(n) \right)$  
  determines $s(T)$ in polynomial time.  Let $\cal F$ be an instance of {\sc 3-SAT} with $n$ variables and \AY{$m$} clauses and let $T'$ be the instance of {\sc Feedback vertex set} for tournaments on $n'=9m$ vertices that we obtain by the sequence of transformations above. 
 Let \AY{$N= \lceil 2^{(n')^{\frac{1}{1+\epsilon}}} \rceil$, which implies that $n'$ is very close to (and at most) $\log^{1+\epsilon}(N)$.}

  Now we construct a new (non-strong) tournament $T$ from $T'$ by adding the vertices of a transitive tournament $T^*$ on
  $N-n'$ vertices and adding all possible arcs from $V(T')$ to $V(T^*)$. As in the proof of Theorem \ref{thm:NPhard_tournament}, $T$ has a safe set of size $k+1$ if and only if $T'$ has a feedback vertex set of size $k$. Hence by performing the transformation of $\cal F$ to $T'$ and then to \AY{$T^*$} we can solve $\cal F$ by checking whether \AY{$T^*$} has a safe set of size at most $2m+1$. 
Thus if  $\cal A$ exists and has running time $O(N^c)$ for some constant $c$, then we can solve \AY{{\sc 3-SAT}} in time 
$\AY{O((2^{(n')^{\frac{1}{1+\epsilon}}})^c)=O(2^{c\cdot{}(n')^{\frac{1}{1+\epsilon}}})}$.\\

  Now using that \AY{$n'=9m$} % and $n \leq 3m$ (as $\cal F$ is an instance of {\sc 3-SAT})}
 we obtain the following.

\AY{
  \begin{eqnarray*}
    c\cdot{}(n')^{\frac{1}{1+\epsilon}}&=& c\cdot{}(9m)^{\frac{1}{1+\epsilon}}\\
                                     &=& c\cdot{}9^{\frac{1}{1+\epsilon}}\cdot{} m^{\frac{1}{1+\epsilon}}\\
    &<& c\cdot{}9^{\frac{1}{1+\epsilon}}\cdot{}(n+m)^{\frac{1}{1+\epsilon}}\\  
  \end{eqnarray*}

  Given any $\delta >0$ we note that when $m$ gets large enough we have $c\cdot{}9^{\frac{1}{1+\epsilon}}\cdot{}(n+m)^{\frac{1}{1+\epsilon}} < \delta (n+m)$. 
Therefore  the algorithm we designed above solves  {\sc 3-SAT} faster than $O(2^{\delta{}(n+m)})$, contradicting ETH. }
    \end{proof}

 }

As the proof for the following theorem is identical if we consider tournaments or semicomplete digraphs\footnote{A digraph $D$ is semicomplete if there is at least one arc between any pair of distinct vertices.}, we will prove it for the slightly 
larger class of semicomplete digraphs. Let $lsc(D)$ {denote the size of} a largest strong component of a semicomplete digraph $D$.

\begin{theorem}\label{thm:small_strong_comp_poly_tournament}  

There exists an algorithm that finds a smallest safe set in semicomplete digraphs, $D$ \JBJ{in time} 
$O(|V(D)|^2  + |V(D)| \times lsc(D) \times 2^{lsc(D)})$.
\end{theorem}

\begin{corollary}
For every $c>0$, there exists a polynomial time algorithm that can determines $s(T)$ when $T \in {\cal T}(c \log (n))$.
\end{corollary}

\begin{proof}
 Let $T \in  {\cal T}(c \log (n))$  and let $lsc(T)$ be the size of largest strong component of $T$.
By the definition of ${\cal T}(c \log (n))$ we note that $lsc(T) \leq c \log (n)$. By 
Theorem~\ref{thm:small_strong_comp_poly_tournament}  there therefore exists a 
$O(|V(T)|^2  + |V(T)| \times lsc(T) \times 2^{lsc(T)}) = O(n^2  + n \times c \log(n) \times 2^{c \log (n)}) = O(n^2 + c \log(n) \times n^{1+c})$ algorithm,
which is polynomial for any constant $c$.
\end{proof}

%\subsection{Proof of Theorem~\ref{thm:NPhard_tournament}}

%\subsection{Proof of Theorem~\ref{thm:small_strong_comp_ETH_tournament}}

\subsection{Proof of Theorem~\ref{thm:small_strong_comp_poly_tournament}}

Let $lsc(D)$ be the size of the largest strong component of a semicomplete digraph $D$. We will now give a 
$O(|V(D)|^2  + |V(D)| \times lsc(D) \times 2^{lsc(D)})$ algorithm that determines $s(D)$.
Let $C_1,C_2,\ldots, C_p$ be the strong components of $D$ and without loss of generality assume that all arcs between $C_i$ and $C_j$ go from $C_i$ to $C_j$ if and only if $i<j$. 
We will now determine $s(D)$ using dynamic programming.

Let $S^*(a,b)$ determine a smallest set of vertices from $V_a = \cup_{i=a}^p V(C_i)$ that forms a safe set in $\induce{D}{V_a}$ such 
that the smallest strong component in $\induce{D}{S^*(a,b)}$ has size $b$. Let $s^*(a,b) = |S^*(a,b)|$. We will now illustrate
how to compute $S^*(a,b)$ for all $a=p,p-1,p-2,\ldots,1$ (in that order) and $b\in \{1,2,\ldots,lsc(D)\}$.

We first compute $S^*(p,b)$ for all $b\in \{1,2,\ldots,lsc(D)\}$. Initially let $s^*(p,b)= \infty$ (and $S^*(p,b)$ is undefined).
We now consider all subsets $W \subseteq V(C_p)$ and check if $W$ is a safe set for
$C_p$. If it is a safe set, then let the size of the smallest strong component in \induce{C_p}{W} be $q$ and if $|W| < s^*(p,q)$
then let $s^*(p,q) = |W|$ and let $S^*(p,q)=W$.  Once this is done for all $W$ we have the correct values for $s^*(p,b)$ 
(and the correct sets $S^*(p,q)$) for all $b\in \{1,2,\ldots,lsc(D)\}$. If we consider the digraph in Figure~\ref{fig:D},
then we would obtain the following values (note that there is no safe set in $C_4$ where the smallest strong component in 
the safe set has size $2$ or $5$).

\begin{figure}[hbtp]
\begin{center}
\tikzstyle{vertexX}=[circle,draw, top color=gray!5, bottom color=gray!30, minimum size=16pt, scale=0.7, inner sep=0.5pt]
\tikzstyle{vertexY}=[circle,draw, top color=gray!5, bottom color=gray!30, minimum size=20pt, scale=0.7, inner sep=1.5pt]
\tikzstyle{vertexBIG}=[ellipse, draw, scale=0.6, inner sep=3.5pt]
\begin{tikzpicture}[scale=0.47]
\node (a1) at (3.0,5.0) [vertexX] {$a_1$};
\node (a2) at (1.0,2.0) [vertexX] {$a_2$};
\node (a3) at (5.0,2.0) [vertexX] {$a_3$};
\draw (3,0.5) node {$C_1$};
\node (b1) at (12.5,5.0) [vertexX] {$b_1$};
\node (b2) at  (8.0,2.0) [vertexX] {$b_2$};
\node (b3) at (11.0,2.0) [vertexX] {$b_3$};
\node (b4) at (14.0,2.0) [vertexX] {$b_4$};
\node (b5) at (17.0,2.0) [vertexX] {$b_5$};
\draw (12.5,0.5) node {$C_2$};
\node (c1) at (  22,5.0) [vertexX] {$c_1$};
\node (c2) at (20.0,2.0) [vertexX] {$c_2$};
\node (c3) at (24.0,2.0) [vertexX] {$c_3$};
\draw (22,0.5) node {$C_3$};
\node (d1) at (27.0,5.0) [vertexX] {$d_1$};
\node (d2) at (27.0,2.0) [vertexX] {$d_2$};
\node (d3) at (30.0,2.0) [vertexX] {$d_3$};
\node (d4) at (30.0,5.0) [vertexX] {$d_4$};
\draw (28.5,0.5) node {$C_4$};

% \draw [rounded corners] (3.3,1.2) rectangle (4.7,5.8);
%  \draw (5.5,3.5) node {$\cdots$};

\draw [->, line width=0.03cm] (a1) -- (a2);
\draw [->, line width=0.03cm] (a2) -- (a3);
\draw [->, line width=0.03cm] (a3) -- (a1);

\draw [->, line width=0.03cm] (b2) to [out=10, in=170] (b3);
\draw [->, line width=0.03cm] (b3) to [out=10, in=170] (b4);
\draw [->, line width=0.03cm] (b4) to [out=10, in=170] (b5);
\draw [->, line width=0.03cm] (b2) to [out=330, in=210] (b4);
\draw [->, line width=0.03cm] (b3) to [out=330, in=210] (b5);
\draw [->, line width=0.03cm] (b2) to [out=300, in=240] (b5);
\draw [->, line width=0.03cm] (b5) -- (b1);
\draw [->, line width=0.03cm] (b4) -- (b1);
\draw [->, line width=0.03cm] (b1) -- (b2);
\draw [->, line width=0.03cm] (b1) -- (b3);

\draw [->, line width=0.03cm] (c1) -- (c2);
\draw [->, line width=0.03cm] (c3) -- (c1);
\draw [->, line width=0.03cm] (c2) to [out=30, in=150] (c3);
\draw [->, line width=0.03cm] (c3) to [out=210, in=330] (c2);

\draw [->, line width=0.03cm] (d1) -- (d2);
\draw [->, line width=0.03cm] (d2) -- (d3);
\draw [->, line width=0.03cm] (d3) -- (d4);
\draw [->, line width=0.03cm] (d4) -- (d1);
\draw [->, line width=0.03cm] (d1) -- (d3);
\draw [->, line width=0.03cm] (d2) -- (d4);

\end{tikzpicture}
\end{center}
\caption{A semicomplete digraph, $D$, where the arcs between distinct $C_i$'s are not shown, but go from $C_i$ to $C_j$ if and only
if $i<j$. Note that $lsc(D)=5$.}
\label{fig:D}
\end{figure}
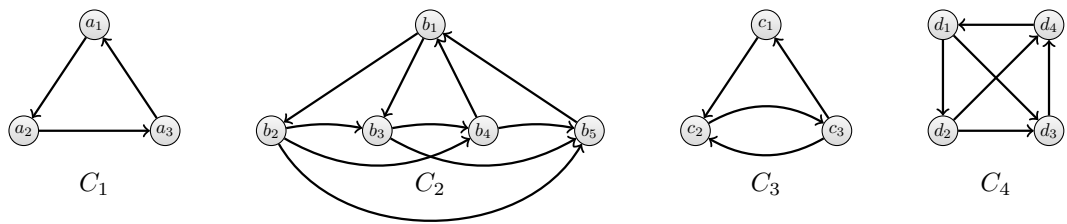

\begin{center}
\begin{tabular}{|c||c|c|} \hline
$b$ & $s^*(4,b)$ & $S^*(4,b)$ \\ \hline \hline
1 & $2$ & $\{d_3,d_4\}$ \\
2 & $\infty$ & - \\
3 & $3$ & $\{d_1,d_2,d_4\}$ \\
4 & $4$ & $\{d_1,d_2,d_3,d_4\}$ \\
5 & $\infty$ & - \\ \hline
\end{tabular}
\end{center}

We now compute $S^*(i,b)$ for all $i=p-1,p-2,\ldots,1$ (in that order) and $b\in \{1,2,\ldots,lsc(D)\}$. 
So when we want to compute $S^*(i,b)$ we may assume that $S^*(i+1,b')$ is known for all $b' \in \{1,2,\ldots,lsc(D)\}$.
Initially let $s^*(i,b)= \infty$ (and $S^*(i,b)$ is undefined).
We now consider all subsets $W \subseteq V(C_i)$, where there is no strong component in $C_i \setminus W$ which has an arc into 
a smaller strong component of \induce{D}{W}. For each such $W$ let $s_W$ be the \JBJ{size of a smallest} strong component in \induce{D}{W} and
let $t_W$ be the \JBJ{size of a largest} strong component in \induce{D}{V(C_i) \setminus W}. For all $j=t_W,t_W+1,\ldots,lsc(D)$ we now perform
the following operation.

\begin{itemize}
\item If $|W|+s^*(i+1,j) < s^*(i,min(j,s_W))$ then,
\begin{description}
\item[(i):]  $S^*(i,min(j,s_W)) = W \cup S^*(i+1,j)$ and 
\item[(ii):]  $s^*(i,min(j,s_W)) = |S^*(i,min(j,s_W))|$.
\end{description}
\end{itemize}

Once this operation is performed for all admissible $W$ we will have the correct values of $S^*(i,b)$ for all  $b\in \{1,2,\ldots,lsc(D)\}$.
This is the case because if we take the union of $W$ and $S^*(i+1,j)$ then the smallest strong component in the resulting set will be the minimum
of $j$ (the smallest strong component in $S^*(i+1,j)$) and $s_W$ (the smallest strong component in $W$). Also the resulting set will be a safe
set in \induce{D}{V(C_i) \cup V(C_{i+1}) \cup \cdots \cup V(C_p)} as  $S^*(i+1,j)$ is a safe set in
\induce{D}{V(C_{i+1}) \cup V(C_{i+2}) \cup \cdots \cup V(C_p)} and every strong component in $C_i \setminus W$ has an arc to all vertices in 
$S^*(i+1,j)$ and is smaller than all strong components in $S^*(i+1,j)$ and it is also smaller than all strong components in $W$ to which it has 
an arc. This proves the correctness of the above procedure.
Again considering the semicomplete digraph in Figure~\ref{fig:D} (and recalling the values $s^*(4,b)$) we would obtain the following values.

\begin{center}
\begin{tabular}{|c||c|c|l} \cline{1-3}
$b$ & $s^*(3,b)$ & $S^*(3,b)$ & \\ \cline{1-3}
1 & $3$ & $\{c_2,d_3,d_4\}$ & Obtained as $\{c_2\} \cup S^*(4,1)$. \\
2 & $5$ & $\{c_2,c_3,d_1,d_2,d_4\}$ & Obtained as $\{c_2,c_3\} \cup S^*(4,3)$. \\
3 & $3$ & $\{d_1,d_2,d_4\}$ & Obtained as $\emptyset \cup S^*(4,3)$.  \\
4 & $4$ & $\{d_1,d_2,d_3,d_4\}$ & Obtained as $\emptyset \cup S^*(4,4)$. \\
5 & $\infty$ & - & \\ \cline{1-3}
\end{tabular}
\end{center}

Continuing the above process we obtain the following values.

\begin{center}
\begin{tabular}{|c||c|c|} \hline
$b$ & $s^*(2,b)$ & $S^*(2,b)$  \\ \hline
1 & $4$ & $\{b_1,c_2,d_3,d_4\} = \{b_1\} \cup S^*(3,1)$. \\
2 & $8$ & $\{b_1,b_2,b_5,c_2,c_3,d_1,d_2,d_4\} = \{b_1,b_2,b_5\} \cup S^*(3,2)$. \\
3 & $6$ & $\{b_1,b_2,b_5,d_1,d_2,d_4\} = \{b_1,b_2,b_5\} \cup S^*(3,3)$. \\
4 & $8$ & $\{b_1,b_2,b_3,b_4,d_1,d_2,d_3,d_4\} = \{b_1,b_2,b_3,b_4\} \cup S^*(3,4)$.  \\
5 & $\infty$ & -  \\ \hline
\end{tabular}
\end{center}

\begin{center}
\begin{tabular}{|c||c|c|} \hline
$b$ & $s^*(1,b)$ & $S^*(1,b)$ \\ \hline
1 & $5$ & $\{a_1,b_1,c_2,d_3,d_4\} = \{a_1\} \cup S^*(2,1)$. \\
2 & $11$ & $\{a_1,a_2,a_3,b_1,b_2,b_5,c_2,c_3,d_1,d_2,d_4\} = \{a_1,a_2,a_3\} \cup S^*(2,2)$. \\
3 & $6$ & $\{b_1,d_1,d_2,d_4\} = \emptyset \cup S^*(2,3)$.     \\
4 & $8$ & $\{b_1,d_1,d_2,d_3,d_4\} = \emptyset \cup S^*(2,4)$.  \\
5 & $\infty$ & -                 \\ \hline
\end{tabular}
\end{center}

We now obtain the value $s(D)$ as the minimum value of $s^*(1,b)$ over all $b=1,2,\ldots,lsc(D)$.  In our above example we get
$s(D)=5$ and a safe set obtaining this value is $S^*(1,1) = \{a_1,b_1,c_2,d_3,d_4\}$. The correctness of the above algorithm \JBJ{follows from the fact that} the values of $s^*(i,b)$ and the sets $S^*(i,b)$ are computed correctly in each step.

We will now compute the complexity of the algorithm. Let $n$ denote the order of $D$. 
 Note that, using e.g. depth-first search, we can find the components $C_i$ in $O(n^2)$ time and in order to compute $s^*(i,b)$ of a given value of $i$ (and all $b$) 
we consider all $W \subseteq V(C_i)$. There are at most $2^{lsc(D)}$ such sets. For each of these sets we may need to consider 
all $j=t_W,t_W+1,\ldots,lsc(D)$ and there are at most $lsc(D)$ such $j$-values. Therefore we can compute all $s^*(i,b)$ for a given $i$
(and every $b$) in time $O(lsc(D) \times 2^{lsc(D)})$.  As there are at most $V(D)$ values of $i$ we get a total 
complexity of $O(n^2  + n \times lsc(D) \times 2^{lsc(D)})$, as desired.
This completes the proof of Theorem~\ref{thm:small_strong_comp_poly_tournament}.

\section{Safe number and strong safe number of tournaments}\label{tournament}

As we have seen so far, finding a minimum safe set in a digraph $D$ is a difficult problem even if we restrict $D$ \JBJ{to be a traceable acyclic digraph or a tournament}. 
In this section, we consider \JBJ{bounds on the safe number and the strong safe number for tournaments}. 

Let $\mathcal{T}_{n}^{k}$ be the set of $n$-vertex tournaments $T$ with
$\kappa{}(T)=k$.
Below we consider the safe number and the strong safe number of tournaments in $\mathcal{T}_{n}^{k}$.
Let
$$
s^{min}(\mathcal{T}_{n}^{k})=\min \{s(T): T\in \mathcal{T}_{n}^{k}\},~~
s^{max}(\mathcal{T}_{n}^{k})=\max \{s(T): T\in \mathcal{T}_{n}^{k}\},
$$
$$
\cs{}^{min}(\mathcal{T}_{n}^{k})=\min \{\cs{}(T): T\in \mathcal{T}_{n}^{k}\},~~
\cs{}^{max}(\mathcal{T}_{n}^{k})=\max \{\cs{}(T): T\in \mathcal{T}_{n}^{k}\},
$$

Now we will try to determine these four parameters of tournaments. 

Let $\mathcal{T}^{*}$ be the set of $n$-tournaments $T$
whose connectivity is $\lfloor\frac{n-1}{2}\rfloor$. It follows from proposition \ref{kappaatmost} that $\mathcal{T}^{*}\neq \emptyset$.
The following  result will be used later.

%\begin{lemma}[see \cite{HM1966}]
%\label{lemma: every strong tournament is hamiltonian}
%Every strongly connected tournament has a Hamilton cycle.
%\end{lemma}

\begin{theorem}[Reid \cite{Reid1985} and Song \cite{Song1993}]
\label{lemma: cycle-factors in 2-connected tournaments}
\JBJ{Every $2$-strong tournament $T$ which is not isomorphic to the unique $7$-tournament containing no transitive subtournament on 4 vertices has a cycle-factor consisting of two cycles of cycle lengths $t$ and $|T|-t$ for any $3\leq t\leq |T|-3$.
}\end{theorem}

\begin{theorem}\label{proposition: minimality of safe numbers}
For every non-negative integer $k$ we have\\

$s^{min}(\mathcal{T}_{n}^{k})=
\begin{cases}
k+1,& \text {$0\leq k\leq 2$};\\
k,& \text {$3\leq k\leq \lceil n/2 \rceil -1$}.
\end{cases}
$

$
\cs{}^{min}(\mathcal{T}_{n}^{k})=
\begin{cases}
1,& \text {$k=0$};\\
3,& \text {$k=1,2$};\\
k,& \text {$3\leq k\leq \lceil n/2 \rceil -1$}.
\end{cases}
$
\end{theorem}

\begin{proof}
Clearly, $\cs{}^{min}(\mathcal{T}_{n}^{k})\geq s^{min}(\mathcal{T}_{n}^{k})\geq 1$.
If a tournament $T$ has a (strong) safe set of cardinality one, say $\{v\}$,
then $T\backslash v$ is transitive and $T\backslash v\rightarrow v$.
It follows that $T$ is transitive.
The converse statement holds too.
This implies that
a tournament has a (strong) safe set of cardinality one
if and only if it is transitive.
So $s^{min}(\mathcal{T}^{k}_{n})\geq 2$
and $\cs{}^{min}(\mathcal{T}^{k}_{n})\geq 3$ for $k\geq 1$.

We first observe  that
$s^{min}(D)\geq k$ for every $k$-strong oriented graph $D$. This is a simple consequence of the fact that $D-S$ is strongly connected for all subsets $S\subset V$ with $|S|<k$ and the fact that 
 there exists an arc from $V\backslash S$ to $S$:
%By Lemma \ref{lemma: every strong tournament is hamiltonian},
%$D[V\backslash S]$ has a Hamilton cycle.
If $S$ is a safe set and $k>|S|$ then $D[V\setminus S]$ is strong and has an arc to $S$ so it follows from (ii) that we have $k>|S|\geq |V\setminus S|=n-|S|\geq k+1$, contradiction. Here we used the fact that, as pointed out in Proposition \ref{kappaatmost}, 
every $k$-strong oriented  graph has at least $2k+1$ vertices.
So $s^{min}(\mathcal{T}_{n}^{k})\geq k$.
Also, one can see that any safe set of cardinality $k$ of a \JBJ{$k$-strong} tournament $T$
must be a separator  of $T$.

\begin{case}
$k=0$.
\end{case}

Consider a transitive $n$-vertex tournament,
the unique sink forms both a safe set and a strong safe set.
So $s^{min}(\mathcal{T}_{n}^{0})=\cs{}^{min}(\mathcal{T}_{n}^{0})=1$.

\begin{case}
$k=1$.
\end{case}

Consider a tournament $T'$ with vertex set $\{v_{1},\ldots,v_{n}\}$
and $v_{i}\rightarrow v_{j}$ for all $i<j$ with only one exception
$v_{n}\rightarrow v_{1}$.
One can see that $T'$ has a Hamilton cycle $v_{1}\rightarrow v_{2}\rightarrow \ldots\rightarrow v_{n}\rightarrow v_{1}$
and a cut vertex $v_{n}$.
So $T'\in \mathcal{T}^{1}_{n}$.
Clearly, $\{v_{1},v_{n}\}$ is a safe set
and $\{v_{1},v_{2},v_{n}\}$ is a strong safe set of $T'$.
So $s^{min}(\mathcal{T}^{1}_{n})=2$ and $\cs{}^{min}(\mathcal{T}^{1}_{n})=3$.

\begin{case}
$k=2$.
\end{case}

Assume that $s^{min}(\mathcal{T}^{2}_{n})=2$.
Then there exists a tournament $T''$ in $\mathcal{T}^{2}_{n}$
which contains a safe set $S$ of cardinality 2.
Since $T''[S]$ is transitive,
we have that $T''[V\backslash S]$ is also transitive.
Let $v$ be the unique source of $T''[V\backslash S]$.
Since $T''$ is 2-strong,
we have $\delta^{-}(T'')\geq 2$ and $S\rightarrow v$.
There exists no arc from $v$ to $S$,
a contradiction to the definition of a safe set.
So $s^{min}(\mathcal{T}^{2}_{n})\geq 3$
and $\cs{}^{min}(\mathcal{T}^{2}_{n})\geq s^{min}(\mathcal{T}^{2}_{n})\geq 3$.

Now we construct a tournament $T'''$ in $\mathcal{T}^{2}_{n}$
with a safe set of cardinality 3.
Let $V(T''')=\{v_{1},\ldots,v_{n}\}$
and $A(T''')=\{v_{1}v_{2},v_{2}v_{3},v_{3}v_{1}\}
\cup \{v_{i}v_{j}: 4\leq j< i\leq n\}
\cup \{v_{3}v_{4},v_{4}v_{1},v_{4}v_{2}\}
\cup \{v_{i}v_{3},v_{1}v_{i},v_{2}v_{i}: 5\leq i\leq n\}$.
One can check that $T'''$ is 2-strong
and has a vertex-cut set $\{v_{1},v_{2}\}$.
So $T'''\in \mathcal{T}^{2}_{n}$.
Here $T'''$ has a strong safe set $\{v_{1},v_{2},v_{3}\}$.
It follows that $\cs{}^{min}(\mathcal{T}^{2}_{n})=3$
and $s^{min}(\mathcal{T}^{2}_{n})=\cs{}^{min}(\mathcal{T}^{2}_{n})=3$.

\begin{case}
$3\leq k\leq \lceil n/2 \rceil -1$.
\end{case}

Let $T^{k}_{2k+1}$ be a tournament
with $V(T^{k}_{2k+1})=\{v_{0},\ldots,v_{2k}\}$
and $A(T^{k}_{2k+1})=\{v_{i}v_{j}: 0\leq i\leq 2k,
i+1\leq j\leq i+k+1, j\neq i+k\}$,
here the addition is modulo $2k+1$.

We claim that $T^{k}_{2k+1}\in \mathcal{T}^{k}_{2k+1}$.
It is easy to see that $T^{k}_{2k+1}$
has a vertex cut set $\{v_{1},\ldots,v_{k-1},v_{k+1}\}$ of cardinality $k$.
Now we show that it is $k$-strong,
i.e., it remains strongly connected
after deleting any set of fewer than $k$ vertices.

Let $M$ be an arbitrary subset of $V(T^{k}_{2k+1})$ with $|M|=k-1$.
Denote by $V(T^{k}_{2k+1})\backslash M=\{v_{t_{0}},v_{t_{1}},\ldots,v_{t_{k+1}}\}$
with $0\leq t_{0}<t_{1}<\cdots <t_{k+1}\leq 2k$.
It suffices to show that there exists a path
from $v_{t_{i}}$ to $v_{t_{i+1}}$ for any $0\leq i\leq k+1$.
Since $|M|=k-1$,
we have either $t_{i+1}\in \{t_{i}+1,\ldots,t_{i}+(k-1)\}$ or $t_{i+1}=t_{i}+k$.
If $t_{i+1}\in \{t_{i}+1,\ldots,t_{i}+(k-1)\}$,
then $v_{t_{i}}\rightarrow v_{t_{i+1}}$ is a path from $v_{t_{i}}$ to $v_{t_{i+1}}$.
If $t_{i+1}=t_{i}+k$,
then \JBJ{$M=\{v_{t_i+1},\ldots{},v_{t_i+k-1}\}$ and we have}
$v_{t_{i}}\rightarrow v_{t_{i}+k+1}\rightarrow v_{t_{i}+2k}=v_{t_{i}-1}
\rightarrow v_{t_{i}+k}$ is a path from $v_{t_{i}}$ to $v_{t_{i+1}}$.

Let $S=\{v_{1},\ldots,v_{k-1},v_{k+1}\}$
and $S'=\{v_{k},v_{k+2},\ldots,v_{2k}\}$.
Then $V(T^{k}_{2k+1})=\{v_{0}\}\cup S\cup S'$,
$v_{0}\rightarrow S$ and $S'\rightarrow v_{0}$.
Note that both $T^{k}_{2k+1}[S]$ and $T^{k}_{2k+1}[S']$ are strongly connected.
In fact, $T^{k}_{2k+1}[S]$ has a Hamilton cycle
$v_{1}v_{2}\cdots v_{k-1}v_{k+1}v_{1}$
and $T^{k}_{2k+1}[S']$ has a Hamilton cycle
$v_{k}v_{k+2}\cdots v_{2k-1}v_{2k}v_{k}$.
Also, note that \JBJ{the arc} $v_{2k}v_{1}$ \JBJ{goes} from $S'$ to $S$.
\JBJ{Thus} $S=\{v_{1},\ldots,v_{k-1},v_{k+1}\}$
is a \JBJ{strong} safe set of cardinality $k$ in $T^{k}_{2k+1}$.

For any $n\geq 2k+2$,
we can construct a tournament $T^{k}_{n}$ from $T^{k}_{2k+1}$ as follows.
Let $V(T^{k}_{n})=V(T^{k}_{2k+1})\cup V^{0}$,
here $V^{0}=\{u_{1},\ldots,u_{n-2k-1}\}$.
Let $A(T^{k}_{n})=A(T^{k}_{2k+1})
\cup \{u_{i}u_{j}: 1\leq i<j\leq n-2k-1\}
\cup \{uv: u\in V^{0}, v\in \{v_{0}\}\cup S\}
\cup \{wu: w\in S', u\in V^{0}\}$.
It is not difficult to check that $T^{k}_{n}\in \mathcal{T}^{k}_{n}$
and $S$ is a strong safe set.
So $\cs{}^{min}(\mathcal{T}_{n}^{k})=k$ for $3\leq k\leq \lceil n/2 \rceil -1$.
\end{proof}

\begin{theorem}\label{proposition: maximality of safe numbers}
If $n\geq 3k\geq 9$,
then $\lfloor n/2 \rfloor \leq s^{max}(\mathcal{T}_{n}^{k})\leq \cs{}^{max}(\mathcal{T}_{n}^{k})\leq \lceil n/2 \rceil$.
\end{theorem}

\begin{proof}
By definition,
$s^{max}(\mathcal{T}_{n}^{k})\leq \cs{}^{max}(\mathcal{T}_{n}^{k})$.
It suffices to show that $\cs{}^{max}(\mathcal{T}_{n}^{k})\leq \lceil n/2 \rceil$
and $s^{max}(\mathcal{T}_{n}^{k})\geq \lfloor n/2 \rfloor$.
Let $T$ be an arbitrary tournament in $\mathcal{T}^{k}_{n}$ with $k\geq 3$.
By \JBJ{Theorem} \ref{lemma: cycle-factors in 2-connected tournaments},
$T$ has \JBJ{a pair of disjoint cycles} $C_{1}$ and $C_{2}$
with $|C_{1}|=\lceil n/2 \rceil$
and $|C_{2}|=\lfloor n/2 \rfloor$.
Clearly, $V(C_{1})$ is a strong safe set of $T$.
So $\cs{}^{max}(\mathcal{T}_{n}^{k})\leq \lceil n/2 \rceil$.

Now we show that $s^{max}(\mathcal{T}_{n}^{k})\geq \lfloor n/2 \rfloor$ for general $n$.
For even $n$ and odd $n$,
we construct $T^{*}$ and $T^{**}$ respectively satisfying the desired result.

We first construct a tournament $T^{\dag}$ with $|T^{\dag}|=2k'+1$
and satisfying that every safe set contains at least $k'+1$ vertices.
Let $T^{\dag}$ be \JBJ{the regular} tournament
with vertex set $\{v_{0},v_{1},\ldots,v_{2k'}\}$
and arc set $\{v_{i}v_{j}: 0\leq i\leq 2k', j=i+1,\ldots,i+k'\}$,
where addition is modulo \JBJ{$2k'+1$}.
Note that $T^{\dag}$ is $k'$-regular, i.e.,
every vertex has both out-degree and in-degree $k'$.

\begin{claim}\label{claim: constructed tournament}
The connectivity of $T^{\dag}$ is $k'$
and the indices of the vertices in each vertex cut set of cardinality $k'$ form a consecutive sequence.
\end{claim}

\begin{proof}
One can see that $\{v_{1},\ldots,v_{k'}\}$ is clearly a vertex cut set of cardinality $k'$ of $T^{\dag}_{2k'+1}$.
Now we show that $T^{\dag}$ is strongly connected after deleting any set of $k'-1$ vertices.
Let $M$ be an arbitrary subset of $V(T^{\dag})$ with $|M|=k'-1$.
Let $V(T^{\dag}-M)=\{v_{t_{0}},v_{t_{1}},\ldots,v_{t_{k'+1}}\}$.
Assume without loss of generality that $0\leq t_{0}<t_{1}<\cdots <t_{k'+1}\leq 2k'$.
Since $|M|=k'-1$,
we have $t_{i+1}\in \{t_{i}+1,\ldots,t_{i}+k'\}\backslash M$
and $v_{t_{i}}\rightarrow v_{t_{i+1}}$ for each $0\leq i\leq k'+1$, here $v_{t_{k'+2}}=v_{t_{0}}$.
So $T^{\dag}-M$ has a Hamilton cycle
$v_{t_{0}}\rightarrow v_{t_{1}}\rightarrow \ldots \rightarrow v_{t_{k'+1}}\rightarrow v_{t_{0}}$
and is strongly connected.
So $T^{\dag}$ has connectivity $k'$.

For any subset of $k'$ vertices whose indices do not form a consecutive sequence,
by using similar analysis above,
we can show that the removal of these $k'$ vertices results in a Hamiltonian subtournament, that is, strongly connected.
Then the result follows directly.
\end{proof}

Let $S^{\dag}$ be a safe set of $T^{\dag}$ with minimum cardinality \JBJ{and let $\overline{S^{\dag}}$ be its complement}.
If $|S^{\dag}|\leq k'-1$,
then $|\overline{S^{\dag}}|\geq k'+2$
and $T^{\dag}[\overline{S^{\dag}}]$ is strongly connected since $T^{\dag}$ is $k'$-strong.
Recall that every vertex in $S^{\dag}$ has in-degree $k$ and $|S^{\dag}|\leq k'-1$,
so there exists an arc from $\overline{S^{\dag}}$ to $S^{\dag}$.
By the definition of safe sets,
$|\overline{S^{\dag}}|\leq |S^{\dag}|$, a contradiction.

If $|S^{\dag}|=k'$,
then $|\overline{S^{\dag}}|=k+1$.
One can see that there exists an arc from $\overline{S^{\dag}}$ to $S^{\dag}$.
So $T^{\dag}[\overline{S^{\dag}}]$ is not strongly connected by the fact that $S^{\dag}$ is a safe set.
This implies that $S^{\dag}$ is vertex cut set
and, by Claim \ref{claim: constructed tournament},
the indices of vertices in $S^{\dag}$ form a consecutive sequence.
We may assume that $S^{\dag}=\{v_0,...,v_{k'-1}\}$ and $\overline{S}=\{v_{k'},...,v_{2k'}\}$.
Then each vertex of $\overline{S^{\dag}}$ forms a \JBJ{strong} component
and there is no arc from the \JBJ{strong} component $v_{k'}$ to any vertex of $S^{\dag}$.
This contradicts the assumption that $S^{\dag}$ is a safe set.
Thus $|S^{\dag}|\geq k'+1$.

For even $n$,
we may assume that $n=2k'+2\geq 2k+2$
and construct a tournament $T^{*}$ from $T^{\dag}$ by adding a new vertex $v^{*}$
together with $k$ arcs from $v^{*}$ to $\{v_{0},\ldots,v_{k-1}\}$
and $n-k-1$ arcs from $V(T^{\dag})-\{v_{0},\ldots,v_{k-1}\}$ to $v^{*}$.
One can see that $|T^{*}|=2k'+2$ and the connectivity of $T^{*}$ is $k$.
Now $\{v_{0},v_{1},\ldots,v_{k'-1},v_{k'+1}\}$ is clearly a safe set of $T^{*}$ of cardinality $k'+1=n/2$.
We show that $s(T^{*})\geq k'+1=\lfloor n/2\rfloor$,
which will imply that $s^{max}(\mathcal{T}_{n}^{k})\geq \lfloor n/2\rfloor$ for even $n$.

Suppose the opposite that $S$ is a safe set of cardinality at most $k'$ in $T^{*}$.
Since $T^{\dag}$ is $k'$-strong,
we have \JBJ{$S\subset V(T^{\dag})$, $S$ must be } a vertex cut set of $T^{\dag}$ and $|S|=k'$;
otherwise, $T^{\dag}-S$ is strongly connected
and the largest strongly connected component of $T^{*}-S$
has order at least $|T^{\dag}-S|>|S|$, a contradiction.
By Claim \ref{claim: constructed tournament},
the indices of the vertices in $S$ form a consecutive sequence.
We may assume that  $S=\{v_0,...,v_{k'-1}\}$ and $\overline{S}=\{v_k',...,v_{2k'},v^{*}\}$ in $T^{*}$.
Now each vertex of $S$ forms a component in $T^{*}[S]$.
It follows that each vertex in $\overline{S}$ form a component in $T^{*}[\overline{S}]$
and has an out-neighbor in $S$,
a contradiction to the fact that there is no arc from $v_{k'}$ to $S$.

For odd $n$,
if $n=2k+1$
then let $k'=k$ and now $T^{\dag}$ is a tournament with safe number at least $k+1=\lfloor n/2 \rfloor +1$,
satisfies the desired inequality.
It suffices to consider the case of $n>2k+1$.
Assume without loss of generality that $n=2k'+3$
and construct a tournament $T^{**}$ from $T^{*}$ by adding a new vertex $v^{**}$
together with $k$ arcs from $v^{**}$ to $\{v_{0},\ldots,v_{k-1}\}$
and $n-k-1$ arcs from $V(T^{*})-\{v_{0},\ldots,v_{k-1}\}$ to $v^{**}$.
One can see that $|T^{**}|=2k'+3$ and the connectivity of $T^{**}$ is $k$.
Similar to the proof for even $n$,
we can show that $s(T^{**})\geq \lfloor n/2\rfloor$ and $s^{max}(\mathcal{T}_{n}^{k})\geq \lfloor n/2\rfloor$ for odd $n$.
\end{proof}

%\section{Problems to consider}

%\begin{itemize}
%\item Is in-dominating set polynomial or NP-hard for tournaments?
%It is known \cite{megiddoT\CS{}61} that it can be solved in time $O(n^{O(\log{}n)})$.
%\item Does high minimum degree in the underlying digraph $UG(D)$ imply that in-dominating set is polynomial for acyclic digraphs?
%\item What is the complexity of safe set for tournaments?
%\item What can we say about minimal safe sets, that is, safe sets so that no proper subset is also a safe set?
%\item Is $s_{max}(\mathcal{T}^k_n)=\lfloor{}n/2\rfloor{}$?
%\end{itemize}

\section{Some remarks and open questions}\label{last}

\iffalse
As we know, the definition of a safe set coincides with that of an in-dominating set for acyclic digraphs. 
Clearly, the situation is different when we consider safe sets in a non-acyclic digraph, as it contains a strongly connected component with more than one vertex. \fi

As introduced in Section~\ref{tournament}, we can similarly define $s^{max}(\mathcal{D})$ or $\cs{}^{max}(\mathcal{D})$  (that is, $\cs{}^{max}(\mathcal{D})=\max\{\cs{}(D): D\in \mathcal{D}\})$
 for a set of digraphs $\mathcal{D}$ with certain properties. 
Perhaps it would be interesting to investigate these parameters for dense strong digraphs. 
Let $\mathcal{D}_n$ be a set of $n$-vertex digraphs such that 
$\min\{\delta^{+}(D),\delta^{-}(D)\}\geq \lceil n/2 \rceil$ for every $D\in \mathcal{D}_n$.
Note that any digraph $D\in \mathcal{D}_n$ is strongly connected by the following \JBJ{ well-known fact which we prove for the convenience of the reader}.

\begin{fact}
Let $D$ be an $n$-vertex digraph with
$\min\{\delta^{+}(D),\delta^{-}(D)\}\geq \lfloor n/2 \rfloor$.
Then $D$ is strongly connected.
Moreover, the degree condition is best possible in general.
\end{fact}

\begin{proof}
Assume the opposite that $D$ is not strongly connected.
Then $D$ contains a component $D_{1}$
such that there is no arc from $V(D)\backslash V(D_{1})$ to $V(D_{1})$.
Also, $D$ contains a component $D_{2}$
such that there is no arc from $V(D_{2})$ to $V(D)\backslash V(D_{2})$.
Assume without loss of generality that $|V(D_{1})|\leq |V(D_{2})|$.
Then $|V(D_{1})|\leq \lfloor n/2 \rfloor$.
Now we have $d^{-}_{D}(v)\leq \lfloor n/2 \rfloor-1$
for every vertex $v\in V(D_{1})$,
a contradiction.

Now we show that the degree condition is best possible.
Let $D_{1},D_{2}$ be two complete symmetric digraph on $2k+1$ vertices.
Let $D^{*}$ be a digraph with
$V(D^{*})=V(D_{1})\cup V(D_{2})$
and $A(D^{*})=A(D_{1})\cup A(D_{2})\cup \{uv:u\in V(D_{1}),v\in V(D_{2})\}$.
One can see that $\delta^{+}(D^{*})=\delta^{-}(D^{*})=2k=\lfloor |D^{*}|/2 \rfloor -1$
and $D^{*}$ is not strongly connected.
\end{proof}

%Similarly, we can get the following fact.

%\begin{fact}
%Let $T$ be an $n$-vertex tournament.
%If $\min\{\delta^{+}(D),\delta^{-}(D)\}\geq \lceil n/4 \rceil$,
%then $T$ is strongly connected.
%Moreover, the degree condition is best possible in general.
%\end{fact}

A digraph is {\it symmetric}
if $u\rightarrow v$ and $v\rightarrow u$
for every two vertices $u,v$ of the digraph.

\begin{theorem}[H\"{a}ggkvist and Thomassen \cite{HT1976}]
\label{HT}
Let $D$ be an $n$-vertex digraph with
$\min\{\delta^{+}(D),\delta^{-}(D)\}\geq \lceil n/2 \rceil$.
Then $D$ is either pancyclic
or $n$ is even and $D$ is the complete bipartite symmetric digraph \JBJ{$\stackrel{\leftrightarrow}{K}_{n/2,n/2}$}.
\end{theorem}

Note that the complete bipartite symmetric digraph \JBJ{$\stackrel{\leftrightarrow}{K}_{n/2,n/2}$}
has a cycle of length $\lceil n/2 \rceil$.
Therefore, in view of Theorem~\ref{HT}, we have the following fact.

\begin{fact}
$\cs{}^{max}(\mathcal{D}_n)\leq \lceil n/2 \rceil$.
\end{fact}

We propose the following question. 
\begin{question}
Is $\cs{}^{max}(\mathcal{D}_n)= \lfloor n/2 \rfloor$ ?
\end{question}

A digraph $D$ is {\it clique acyclic} if $D$ contains no directed triangle.  
Gy\'arf\'as et al. \cite{GST} strengthened Theorem~\ref{indep} as follows.   

\begin{theorem}[Gy\'arf\'as et al. (Theorem 3; \cite{GST})]
\label{indep2}
Let $f$ be a recursive function such that $f(1)=1$ and for $\alpha\ge 2$, $f(\alpha)=\alpha+\JBJ{\alpha}f(\alpha-1)$. 
If $D$ is a clique acyclic digraph with independence number $\alpha$, then $\gamma(D)\le f(\alpha)$. 
\end{theorem}

In view of Theorem~\ref{indep2} together with the argument in the proof of Theorem~\ref{alpha}, we propose the following conjecture. 

\begin{conjecture}
There exists a polynomial time algorithm to find a minimum safe set in a clique acyclic digraph with a constant independence number $\alpha$. 
\end{conjecture}

\JBJ{In the proof of Theorem \ref{thm:NPhard_tournament} we explicitly use the fact that the tournament $T'$ is not strongly connected so the proof does not extend to the case of strong tournaments.
  \begin{conjecture}
    {\sc safe set} is NP-complete for strong tournaments.
  \end{conjecture}

}

%   A digraph is {\it locally semicomplete} if the out-neighbourhood of every vertex and the in-neighbourhood of every vertex induces a semicomplete digraph.
%   The dynamic programming algorithm in the proof of Theorem  \ref{thm:small_strong_comp_poly_tournament} only uses that the input is a digraph for which there is an acyclic ordering $D_1,\ldots{},D_k$ of the strong components and such that $V(D_i)\dom V(D_j)$ when $i<j$. This implies that the theorem also holds for the class of
%   locally semicomplete digraphs where every strong component has size at most $c\log{}n$ where $n$ is the order of the input.}

%It seems that we can replace $\lceil n/2 \rceil$ by $\lfloor n/2 \rfloor$ in the above fact
%by using the following result. ({\color{red}Not proved yet!})

%\begin{lemma}[Darbinyan \cite{Darbinyan1986}]
%Let $D$ be an $n$-vertex digraph, $n\geq 10$,
%with minimum degree at least $n-1$
%and $\min\{\delta^{+}(D),\delta^{-}(D)\}\geq n/2-1$.
%Then $D$ is pancyclic, unless some extremal cases.
%\end{lemma}

\end{spacing}


\begin{thebibliography}{0}

\bibitem{safe2018}
R. \'{A}gueda, N. Cohen, S. Fujita, S. Legay, Y. Manoussakis, Y. Matsui, L. Montero, R. Naserasr, H. Ono, Y. Otachi, T. Sakuma, Z. Tuza, R. Xu,
Safe sets in graphs: Graph classes and structural parameters. \textit{Journal of Combinatorial Optimization}, \textbf{36} (2018) 1221-1242.




\bibitem{BG2008}
J. Bang-Jensen and G. Gutin,
Digraphs: Theory, Algorithms and Applications,
in: Springer Monographs in Mathematics,
Springer Verlag, London, 2008.


\bibitem{bangSJDM5}
  J. Bang-Jensen and C. Thomassen, A polynomial algorithm for the 2-path problem for semicomplete digraphs,
  \textit{SIAM Journal on  Discrete Mathematics} \textbf{5} (1992) 366-376.


\bibitem{cygan2016}
M. Cygan, F. V. Fomin, L.Kowalik, D. Lokshtanov,
D. Marx, M. Pilipczuk,
M. Pilipczuk and S. Saurabh, Parameterized algorithms, Springer Verlag 2016.
  
  
\bibitem{ehard}
S. Ehard and D. Rautenbach,
Approximating connected safe sets in weighted trees,
\textit{ArXiv}:1711.11412v2 (2017).


\bibitem{wsf}
R. B. Bapat, S. Fujita, S. Legay, Y. Manoussakis, Y. Matsui, T. Sakuma, Z. Tuza,
Weighted safe set problem on trees, \textit{Networks}, \textbf{71} (2018) 81--92.

\bibitem{b}
R. Belmonte, T. Hanaka, I. Katsikarelis, M. Lampis, H. Ono, Y. Otachi,
Parameterized complexity of safe set,
\textit{ArXiv}:1901.09434 (2019).

\bibitem{cormen2009}
  T.H. Cormen, C.E. Leiserson, R.L. Rivest and C. Stein, Introduction to algorithms 3rd edition, MIT press, Cambridge Massachusetts, 2009.

%\bibitem{Darbinyan1986}
%S. Kh. Darbinyan,
%On the pancyclicity of digraphs with large semidegrees,
%Akad. Nauk Armyan. SSR Dokl. {\bf 83} (1986) 99-101.

\bibitem{SF2018}
S. Fujita and M. Furuya,
Safe number and integrity of graphs,
\textit{Discrete Applied Mathematics.}, \textbf{247} (2018)  398--406

\bibitem{FPS:path:cycle}
S. Fujita, T. Jensen, B. Park, T. Sakuma, On weighted safe set problem on paths and cycles, \textit{Journal of Combinatorial Optimization}, \textbf{37} (2019) 685--701.

\bibitem{FMS2016}
S. Fujita, G. MacGillivray, T. Sakuma,
Safe set problem on graphs,
{\it Discrete Applied Mathematics} {\bf 215} (2016) 106-111.

\bibitem{ganianDAM168}
R. Ganian, P. Hlineny, J. Kneis, A. Langer, J. Obdrzalek and P. Rossmanith,
Digraph width measures in parametrized algorithmics, {\em Discrete Applied Mathematics.} {\bf 168} (2014) 88-107.

\bibitem{GST}
A. Gy\'arf\'as, G. Simonyi, \'A. T\'oth, 
Gallai colorings and domination in multipartite digraphs, {\em Journal of  Graph Theory} {\bf 71} (2012) 278-292.

\bibitem{HT1976}
R. H\"{a}ggkvist, C. Thomassen,
On pancyclic digraphs,
{\it Journal of  Combininatorial Theory Ser. B}  {\bf 20} (1976) 20-40.

%\bibitem{HM1966}
%F. Harary and L. Moser,
%The theory of round robin tournaments,
%Amer. Math. Monthly {\bf 73} (1966) 231-246.

\bibitem{host}
P. Hosteins,
A compact mixed integer linear formulation for safe set problems, 
preprint. 

\bibitem{kkp}
B. Kang, S-R. Kim, B. Park, On the safe sets of Cartesian product of two complete graphs, \textit{Ars Combinatoria}, \textbf{141} (2018) 243--257.

\bibitem{karp1972} R.M. Karp, Reducibility among combinatorial problems. In R.E. Miller and J.W. Thatcher (eds.), {\em Complexity of computer computations}, Plenum Press New York (1972) 85-103.

\bibitem{megiddoTCS61} N. Megiddo and U. Vishkin, On finding a mimimum dominating set in a tournament, {\em Theoretical Computer Science} {\bf 61} (1988) 307-316.

\bibitem{Reid1985}
K.B. Reid,
Two complementary circuits in two-connected tournaments,
{\it Annals of Discrete Mathematics} {\bf 27} (1985) 321-334.

\bibitem{Song1993}
Z.M. Song,
Complementary cycles of all lengths in tournaments,
{\it Journal of  Combinatorial Theory Ser. B} {\bf 57} (1993) 18-25.


\bibitem{speckenmeyerLNCS411}
  E. Speckenmeyer, On feedback problems in digraphs, in Proc. 15 WG 89, Springer-Verlag,
  \textit{Lecture Notes in Comput. Science}, \textbf{411} (1989), pp. 218-231.
\iffalse
\bibitem{toveyDAM8}
  C.A. Tovey, A simplified NP-complete satisfiability problem, \textit{Discrete Applied Mathematics,}, \textbf{8} (1984) 85-89.\fi
  
\end{thebibliography}
\end{document}